\newcommand{\version}{November 30, 2007}

\documentclass[12pt]{article}
\usepackage{bbm}
\usepackage{a4,amsthm,amsfonts,latexsym,amssymb}

\swapnumbers 

\theoremstyle{plain}
\newtheorem{thm}{THEOREM}

\newtheorem{lem}[thm]{LEMMA}

\newtheorem{proposition}[thm]{PROPOSITION}

\newtheorem{definition}[thm]{DEFINITION}

\newcommand{\beq}{\begin{equation}}
\newcommand{\eeq}{\end{equation}}
\def\beqa{\begin{eqnarray}}
\def\eeqa{\end{eqnarray}}

\newcommand{\C}{{\mathbb C}}

\newcommand{\R}{{\mathbb R}}

\newcommand{\one}{{\mathbbm 1}}
\newcommand{\Tr}{{\rm Tr}}
\newcommand{\rank}{{\rm rank}}

\newcommand{\s}{{\mathbf{S}}}
\newcommand{\B}{{\mathcal B}}

\newcommand{\Hh}{{\mathcal H}}
\newcommand{\V}{{\mathcal V}}

\newcommand{\D}{{\mathcal D}}
\newcommand{\TT}{{\mathcal T}}
\newcommand{\eps}{\varepsilon}

\newcommand{\half}{\mbox{$\frac{1}{2}$}}

\date{\small\version}

\begin{document}
%%\Large
\markboth{\scriptsize{BNT  \version}}{\scriptsize{BNT
\version}}

\title{
\vspace{-80pt}
\begin{flushright}
{\small UWThPh-2007-25} \vspace{30pt}
\end{flushright}
\bf{Analysis of quantum semigroups with GKS--Lindblad generators I.
\newline
Simple generators  }}
\author{\vspace{8pt} Bernhard Baumgartner$^1$, Heide Narnhofer$^2$, Walter Thirring \\
\vspace{-4pt}\small{Fakult\"at f\"ur Physik, Universit\"at Wien}\\
\small{Boltzmanngasse 5, A-1090 Vienna, Austria}}

\maketitle

%%%\textsl{Draft}\qquad \version

\begin{abstract}
Semigroups describing the time evolution of open quantum systems in finite dimensional spaces have generators
of a special form, known as Lindblad generators.
The simple generators, characterized by only one operator, are analyzed.
The complete set of all the stationary states is presented in detail,
including a formula to calculate a stationary state from the generating operator.
Also the opposite task can be fulfilled, to construct an evolution leading
to a prescribed stationary state.
\\[10ex]
PACS numbers: \qquad  03.65.Yz , \quad 05.40.-a , \quad 42.50.Dv ,  \quad 03.65.Fd
\\[3ex]
Keywords: open system, time evolution, Lindblad generator, semigroup

\end{abstract}

\footnotetext[1]{\texttt{Bernhard.Baumgartner@univie.ac.at}}
\footnotetext[2]{\texttt{Heide.Narnhofer@univie.ac.at}}

%%\maketitle

%%%%%%%%%%%%%%%%%%%%%%%%%%%%%%%%%%%%%%%%%%%%%%%%%%%%%%%%%%%%%%%%%%%%%%%%%%%%%%%%

\section{Introduction}\label{intro}

Complete positivity is an essential
property of maps acting on states of open quantum systems, representing their change in time.
It was in the 1970s, starting with \cite{K71}, when this has been discovered.
In the sequel the generators of semigroups formed by maps with this property
have been identified for finite dimensional systems,
and also for infinite systems, provided a strong continuity property is assumed in addition to complete positivity,
\cite {GKS76, L76}; they are now known as ``Lindblad generators''.
The interest at that time has been mainly in connection with Non Equilibrium Thermodynamics,
see \cite {D76, S80, AF01, BP02} and references therein.

New interest has arisen recently in connection with quantum engineering of small systems,
and it seems necessary to enlarge our knowledge about such semigroups.
So we proceed in the analysis of these completely positive semigroups, acting on states
for systems with \emph{finite dimensional} Hilbert spaces.
We are interested in characterizing both the set of invariant states
and the paths of changing states.
The early results on invariance stated in \cite{S76,F78,S80} concern
systems under special conditions on the set of generators.
We give more general detailed insight, not needing these restricting conditions.
Pioneering work has been done by K. Dietz  \cite{D03,D04,D05},
identifying ``superselection sectors'', discussing the change in time of entropy
and presenting a formula which relates a stationary state to the operator
which appears in a simple generator.
In the present paper his findings are imbedded into a complete mathematical characterization.
In addition we shed some light on the approach to equilibrium,
on the geometry of the paths in the set of states.

There are three types of physical processes one wants to describe:
Decay, dissipation and decoherence.
Decay of excited states leads -- in an idealized model -- to a complete
annihilation of some subspaces of the Hilbert space.
Dissipation, on the other hand, leads to a spreading of the state over a large space,
getting more mixed, in case it started as a pure state.
It is somehow the counterpart to decay and also to conservation laws.
In reality and also in our studies both of these effects occur simultaneously,
and conservation laws may impose some structuring.
Decoherence in a wide sense is dephasing;
in a preferred basis the off-diagonal elements of the density matrix are diminishing
and finally vanishing.
Each type of these processes has its mathematical expression in properties of
the quantum dynamical semigroup's generator.

In this paper we concentrate on simple generators, the building blocks for
generators of general completely positive semigroups
- a characterization of Complete Positivity is presented in the Appendix.
For two dimensional systems, a.k.a. qubits, we present a detailed description.
In a companion paper we extend our interest to the general quantum dynamical semigroups
on finite dimensional Hilbert spaces.

Our starting point, the connection with the earlier studies,
is the result of \cite {GKS76, L76}:

\begin{proposition}{\bf Generators of semigroups:}\label{lindbladprop}
Every generator of a semigroup of completely positive trace preserving maps
$\TT^t:\, \rho(s)\mapsto\rho(s+t)$ for $t\geq 0$
on the set of finite dimensional density matrices $\rho$, can be written in the form
\beq\label{lindblad}
\dot{\rho}=\D(\rho)= -i[H,\rho]  +\sum_\alpha \D_{h^\alpha}(\rho), \qquad \textrm{where}\quad H=H^\dag,
\eeq

where the dissipative parts, the $\D_{h^\alpha}$, are:
\end{proposition}
\begin{definition}{\bf Simple generators:}\label{simplegenerator}
\beq
\D_h (\rho)=h\rho h^\dag -\half (h^\dag h\rho+\rho h^\dag h)
\eeq
\end{definition}

The commutator with a Hamiltonian alone generates a unitary time evolution, which is well known in principle.
We concentrate on
the other building blocks, the simple generators $\D_h$
and the evolutions they generate.

It is well known that the division of $\D$ into a sum of several simple generators is not unique.
Different sets of operators $h^\alpha$ can be attributed to a given generator $\D$.
The attribution of a single $h$ to a simple generator on the other hand is mostly unique, with few exceptions.
The method of proof for this statement uses a structuring of the Hilbert space
and of the operator $h$, which reflects some characteristic properties of the simple $\D$.
These methods are presented in Section \ref{structure},
At the end of the same Section the mentioned uniqueness is stated in the Lemma \ref{attribute}
and then proved.

Our analysis starts with studying the decompositions of $h$,
and their relations to the action of the generator.
It turns out that the time evolution of $\rho$ in the
convex set $\s$ of density matrices can be analyzed
by studying the action of $\D$ on the boundary of $\s$, which is the set of
density matrices with rank less than the dimension of the Hilbert space.
The structure of the operator $h$ and the corresponding structuring of the
Hilbert space are here of prime importance.
Only a density matrix the support of which is an eigenspace of $h$ does not dissipate
in first order of $t$ into higher dimensional subspaces.
Therefore we call these faces of $\s$ which are formed by the $\rho$
with support being an eigenspace or a generalized eigenspace of $h$
the ``lazy faces''.
An even stronger condition marks certain eigenspaces and the corresponding
faces as enclosures. No path $\rho(t)$ leads into or out of them,
some conservation laws are valid.
A closer look on the action of $\D_h$ on density matrices
in special faces lets us determine precisely the complete set of stable, ``stationary'', states.
These results are presented in Theorem \ref{simplestat}.
A note on naming: In the literature the characterization of a state as ``invariant'', ``stable''
or ``stationary'' is used, each one meaning $\dot{\rho}=0$.
In this paper we favor ``stationary''.
And we are sloppy in not differentiating ``state'' and ``density matrix''.

Further analysis, making a detour using the dual time evolution of the operators
and the Kadison inequality, leads to the insight that there are no
periodic evolutions, no circular paths for the $\rho(t)$.
All paths in the space $\s$ lead to the stationary states.
The entropy of the states along a path may increase or decrease
and may even pass through several local maxima and minima.

Some examples, most of them with qubits, demonstrate the findings
of our analysis.

At various places we have to do calculations concerning a splitting
of the Hilbert space into two orthogonal subspaces.
The elaborated formulas are collected in Appendix \ref{appblock}.

\section{The set of generators}\label{genset}

\subsection{Basic properties of the generator's actions}\label{basic}

By inspection of the definition (\ref{simplegenerator}) it is obvious that
the semigroups generated by the superoperators $\D_h$ preserve the trace and selfadjointness of matrices.
\beq
\frac d{dt}\Tr\rho(t)=\Tr[\D(\rho(t))]=0, \qquad \quad \frac d{dt}\sigma^\dag(t)=\D(\sigma^\dag(t))=[\D(\sigma(t))]^\dag.
\eeq
To see that they also preserve positivity, we look at the parts of the simple $\D_h$ separately.
The first part, mapping a positive $\rho$ to a non-negative operator $h\rho h^\dag$,
causes an increase of $\rho$.
The second part causes a decrease of the eigenvalues.
The argument goes as in proving the Feynman-Hellmann theorem and gives:

\begin{proposition}{\bf A differential inequality:}\label{inequality}
If $r(t)$ is a non-negative eigenvalue of $\rho(t)$, its change in time is bounded from below
as
\beq
\dot{r}(t)\geq - \| h\|^2 r(t).
\eeq
\end{proposition}
\begin{proof}
Differentiating the eigenvalue equation
$\rho(t)\psi(t)=r(t)\psi(t)$,
-- Rellich's theorem guarantees that all terms are analytic functions, --
and using $\langle\psi|\dot{\psi}\rangle=0$
gives \beqa
\dot{r}(t)=\langle\psi|\dot{\rho}|\psi\rangle &=&
\langle\psi|(h \rho h^\dag -
\half(h^\dag h\rho+\rho h^\dag h))|\psi\rangle \nonumber\\
&\geq& -\half\langle\psi|(h^\dag h\rho+\rho h^\dag h)|\psi\rangle \nonumber\\
&=&-\langle\psi|h^\dag h|\psi\rangle r(t).
\eeqa
\end{proof}
This implies
that the positive eigenvalues $r_j$ of the density operator obey in the course of time the inequality
\beq\label{edecrease}
r_j(t)\geq\exp\left(- \| h\|^2 t\right)\cdot r_j(0)>0,
\eeq
and positive eigenvalues stay positive.
By continuity of the time evolution, zero eigenvalues can not get negative,
the positivity of density matrices is preserved.

So the evolutions do not leave $\s$. Moreover, there exist stationary states,
i.e. states which do not change in the course of time.
\begin{proposition}{\bf Existence of stationary states.}\label{existence}
For each simple generator $\D_h$ there exists a stationary state
$\rho_\infty$ satisfying $\D_h(\rho_\infty)=0$.
If $h$ has zero as an eigenvector, each state with support in the zero subspace
is stationary.
If $h$ does not have zero as an eigenvalue, one may construct a stationary state by
\beq
\rho_\infty=(h^\dag h)^{-1}/\Tr[(h^\dag h)^{-1}].
\eeq
\end{proposition}

\begin{proof}
If $\rho =\sum_j r_j|\phi_j\rangle\langle\phi_j|$ where $h|\phi_j\rangle=0$, obviously $h\rho=0$,
$\rho h^\dag=0$ and therefore $\D_h(\rho)=0$ hold.
For $h$ without zero as an eigenvector, consider the polar decomposition
$h=U|h|$, where $|h|=\sqrt{h^\dag h}$ and $U=h|h|^{-1}$.
The evolution equation for $\rho_\infty=|h|^{-2}/\Tr[|h|^{-2}]$ is
$$\dot{\rho}_\infty=U|h|\rho_\infty |h|U^\dag-\half(|h|^2\rho_\infty+\rho_\infty |h|^2)=
(U\cdot U^\dag -\one)/\Tr[|h|^{-2}]=0.$$
\end{proof}

This relation between $h$ and $\rho_\infty$ has already been given in \cite{D03}.
The questions of uniqueness and of attraction to paths of other states
are answered in Theorems \ref{simplestat} and \ref{nocirclethm}.

\subsection{Superoperators}\label{superop}

In the sequel it will be useful to endow the linear space of matrices with an inner product
$\langle \sigma|\rho\rangle=\Tr \sigma^\dag \rho$,
and consider matrices as elements of this space, the \emph{Hilbert-Schmidt, HS space}.
In the middle of the set of density matrices $\s$ is the maximally mixed state $\omega=\one/\dim(\Hh)$.
The square distance $\|\rho-\omega\|^2$ is related to the
``Fermi entropy'' $\Tr\rho(\one-\rho)$, on which we make some remarks in Subsection \ref{mixing}.

Consider a generator $\D$ as a
superoperator acting on the
HS space. In this respect the semigroup which it generates can be
extended to a group. For negative $t$ the trace and selfadjointness
of matrices are still preserved, but not the positivity.

 The \emph{adjoint} $\D^\dag$ can be looked
upon as generator of the \emph{dual time evolution} of observables \emph{in the
Heisenberg picture}.
\beq\label{Heisenberg}
\D^\dag_h (f)=h^\dag f h -\half (h^\dag h f+ f h^\dag h).
\eeq
Preservation of the trace of density operators
corresponds to preservation of the unity under the dual time
evolution of observables, the semigroup generated by $\D^\dag$ is ``unital''.
(That both $\D$ and $\D^\dag$ are unital and trace preserving occurs,
if $h$ is a normal operator.)

Diagonalizability of $\D$ is rather exceptional.
Nevertheless there exists for each $\D$ its set of eigenvalues $\{\lambda\}$, equal to the set of zeros
of the characteristic polynomial ${\det}(\D-\lambda\one)$,
a set of proper eigenmatrices $\{\sigma\}$, obeying $\D\sigma=\lambda\sigma$,
and, in some cases of degeneracy, generalized eigenmatrices, obeying
$(\D-\lambda)^n\sigma=0$.
As in any set of linear differential equations, any $\rho(t)$ can be expanded
into a sum of such $\sigma\cdot e^{\lambda t}\cdot {\rm Polynomial}(t)$.
Some general properties of eigenvalues and eigenmatrices are presented in Subsection \ref{nocircle}.
Examples and special cases will be presented for twodimensional systems in Subsection \ref{twobytwo}.

The generators $\D$, defined in Proposition \ref{lindbladprop}, with $H=0$ form a convex cone.
The simple generators make the extremal rays.
Nevertheless closure adds the generators of unitary groups:

\begin{proposition}\label{unitary}
Generators of unitary time evolution can arise as limit of simple Lindbladian generators.
\end{proposition}
\begin{proof}
Define $h(\lambda):=2\lambda^{-1}\one -i\lambda H$ where $H^\dag=H$. Consider the limit $\lambda \rightarrow 0$:
\beq
\D_{h(\lambda)}(\rho)=-i[H,\rho]+\lambda^2\D_H(\rho) \quad \rightarrow\quad -i[H,\rho]
\eeq
\end{proof}
Note that the norm of $h(\lambda)$ diverges, but $\D_{h(\lambda)}$ converges in norm.

So it is to be expected that a thorough inspection of the extremal rays,
containing the building blocks, the simple generators,
will give a good basis for the analysis of all GKS-Lindblad generators.

\section{Simple generators}

\subsection{Splitting, Decomposition, Structure}\label{structure}

We consider the generators $\D=\D_h$, defined in Definition \ref{simplegenerator}.
To analyze their properties it is convenient
to find structures on two levels.
\emph{The structure in the large} concerns a kind of enclosure:
subspaces of $\Hh$
with the property that the system can neither leave nor enter any of them in the course of time.
(These subspaces are related to the superselection sectors in \cite{D04}.)
$P$ is a projector onto such a subspace if and only if it obeys
\beq
[P,h]=0.
\eeq
Each operator $h$ appearing in a generator $\D_h$ can be decomposed as
\beq
h=\bigoplus_j  h_j
\eeq
with a decomposition of the Hilbert space $\Hh=\bigoplus_j \Hh_j $
into mutually orthogonal subspaces $\Hh_j=P_j\Hh$,
where each $h_j$ acts on $\Hh_j$ and is not further decomposable.

Aim and purpose of this decomposition in the large is to reduce the analysis of
the superoperator's  action on any $\rho$ to the analysis to the split action
on the blocks $P_j\rho P_j$ and  $P_j\rho P_\ell$.
It is:
\beq\label{splitaction1}
\D_h\left(P_j\rho P_j\right)= h_j\rho h_j^\dag-\half[h_j^\dag h_j\rho-\rho h_j^\dag h_j]
=\D_{h_j}\left(P_j\rho P_j\right),
\eeq
for the diagonal blocks in the decomposition of $\rho$ into block matrices $P_j\rho P_\ell$.
The off-diagonal blocks give the phase relations between the enclosure-like subspaces.
\beq\label{splitaction2}
\D_h\left(P_j\rho P_\ell\right)= h_j\rho h_\ell^\dag-\half[h_j^\dag h_j\rho-\rho h_\ell^\dag h_\ell].
\eeq

In case of an appearance of equivalent parts $h_j$
there is the possibility of \emph{unitary reshuffling} the corresponding subspaces:
Collect the set of equivalent parts as $h_1 \ldots h_N$, i.e.
$$\forall\{j,\ell\}\subset \{1\ldots N\}\quad \exists\, \textrm{unitary}\, U_{j\ell}: h_j=U_{j\ell}h_\ell U^\dag_{j\ell}.$$
Represent their sum $h_1\oplus\ldots\oplus h_N$ as a tensor product
$\one_N\otimes h_1$ acting in $\C^N\otimes \Hh_1 $.
The action of a unitary operator $V\oplus\one$ in this tensorial representation gives
a unitary reshuffling of eigenspaces $\Hh_j$.

\begin{proposition}{\bf Splitting of operators and of spaces:}\label{simplesplit}
The splitting $h=\bigoplus_j  h_j$
with a set of orthogonal projectors $P_j$ which commute with $h$ is unique, except for a unitary reshuffling
of eigenspaces $\Hh_j=P_j\Hh$ on which equivalent parts $h_j$ are acting.

The action of the generator splits as
\beq
\D_h\left(P_j\rho P_\ell\right)= P_j\D_{h}(\rho)P_\ell.
\eeq

\end{proposition}

\begin{proof}
By induction, one may find a maximal set of mutually commuting orthogonal projectors $P_j$,
all commuting with $h$:
If there exists one $P$ commuting with $h$, one may start with
the set $\{P,\, \one-P\}$. Then one searches for a further splitting $P=P_2+P_3$,
the $P_j$ commuting with $h$. If such a further splitting exists, one replaces $P$ by $P_2$ and $P_3$.
The same for $\one-P$, and then further splittings of the $P_j$, and so on.
 At the end of this procedure, one
has the splitting $\Hh=\bigoplus_jP_j \Hh$, and
$h=\bigoplus_j P_jh P_j$, which can not further be refined.

Now collect piecewise all the mutually equivalent parts $h_{j_1}, ... h_{j_{n(j)}}$ as
$\one_{n(j)}\otimes h_{j_1}$ acting in $\C^{n(j)}\otimes \Hh_{j_1} $.
All the possible families $\{P_j\}$ of projectors can be observed in this way:
The operator $h$ commutes with each projector onto $\V\otimes\Hh_{j_1}$,
where $\V$ is any subspace of $ \C^{n(j)}$, and $h$ commutes with no other projector.

The split action of the parts on diagonal and off-diagonal blocks of $\rho$, stated in the equations
(\ref{splitaction1}) and (\ref{splitaction2}), is obvious.
\end{proof}
This decomposition is an important step in the analysis of the time evolution.
An algebraists point of view is that the projectors commuting with $h$
create a subalgebra of conserved observables,
$\D^\dag(P_j)=0$. The law of conservation implies that
the probability for the system to remain in an enclosing subspace $P_j\Hh$ is constant in time.
The evolution inside the subspace is determined by the action of $\D_{h_j}$.
Studying the total action of $D_h$ consists of determining the action of single $\D_{h_j}$
with indecomposable $h_j$, and studying the evolution of the phase relations
between different subspaces $\Hh_j$.
The indecomposable parts acting on diagonal blocks of $\rho$
are studied in Sections \ref{boundary} and \ref{invariant},
The evolution of the phase relations,
represented in the off-diagonal blocks $P_j\rho P_\ell$, is discussed in
Sections \ref{invariant} and \ref{phases}.

The \emph{structure on the second level} is related
to representions of an indecomposable $h$
- which may enter in the large as an indecomposable part $h_j$ - in a \emph{Schur triangulated form},
as an upper triangular matrix, see f.e. \cite{L69}:
On the diagonal are the eigenvalues in arbitrary sequence,
below the diagonal there are only zeros. An indecomposable $h$ of dimension greater than $1$ is not diagonalizable,
one has to be aware of the generalized concepts:
The eigenspaces and generalized eigenspaces are subspaces $P\Hh$ with projectors $P$,
with the characterizing property
\beq\label{eigen}
hP=PhP.
\eeq
These $P$ are orthogonal projectors, $P=P^2$; each one-dimensional
$P$ projects onto a proper eigenvector. These eigenvectors are in general not mutually orthogonal.
%%For $\dim(\Hh)=n$  - there are at least $n-1$ (generalized) non trivial eigenspaces.
The eigenspaces and generalized eigenspaces
are spanned by the proper and the generalized eigenvectors, $\psi_{\lambda,j,0}$ and $\psi_{\lambda,j,m}$, of $h$.
$$(h-\lambda)\psi_{\lambda,j,0}=0,\quad(h-\lambda)\psi_{\lambda,j,m}=\psi_{\lambda,j,m-1}$$
Note that a generalized eigenspace containing $\psi_{\lambda,j,m}$ has to contain also all the
$\psi_{\lambda,j,\mu}$ with $\mu<m$.

To be sure about the relevance of the properties of $h$ for the action of $\D_h$
it is helpful to study the uniqueness of attributing $h$ to a given simple generator.
Already here the method of Schur triangulation appears as an important tool.

\begin{lem}\label{attribute}
If it is possible to write a generator $D$ as a simple $\D_h$,
the attribution of $h$ is unique up to a phase factor,
unless $h$ is a ``normal'', i.e. diagonalizable, operator.
\end{lem}
\begin{proof}
The strategy of the proof is to consider the equations $\dot{\rho} = \D_h(\rho)$
for the given $h$, with a set of properly chosen matrices $\rho$
as determining the operator $k$ which shall give the same  $\dot{\rho}$, by $\dot{\rho} = \D_k(\rho)$.
The condition of non-diagonalizability of the given $h$ can hold only for $n=dim(\Hh)\geq 2$.
The proper choice of $\rho$ requires a proper choice of the basis vectors,
giving a triangular matrix representation of $h$
and enabling a proof by induction on the rank of $\rho$.
If the given $h$ has a generalized eigenvector $\psi_{\lambda,1}$, choose the basis for its
decomposition and its Schur triangulation in
such a way that $\psi_{\lambda,0}$ is the first basis vector and
$\psi_{\lambda,1}$ is the second one.
If $h$ has only proper eigenvectors, then, because
of its non-diagonalizability, there must be a non-orthogonal pair of them,
$\psi_{\lambda,0}$ and $\psi_{\mu,0}$.
Choose this pair to form the first two basis vectors: $\psi_{\lambda,0}$ and the normalized
$\psi_{\mu,0} - \gamma\psi_{\lambda,0}$, where
$\gamma=\langle\psi_{\lambda,0}|\psi_{\mu,0}\rangle$.
Then we proceed by induction on the dimension of the space
and use the chosen basis.
We refer to the formulas presented in the Appendix  \ref{appblock}.

For dimension 2
the action of $\D$
on the matrix $\rho=\left( \begin{array}{cc} r & q \\ q^\dag & s \end{array} \right) $
with $q=0$ and $s=0$ determines $c=0$ for the matrix
$k=\left( \begin{array}{cc} a & b \\ c & d \end{array} \right)$.
The action on $\rho$ with $r=q=0$ determines first $|b|$, which is not zero,
due to the proper choice of the basis vectors.
Then via $\dot{q}=a^\dag b r$ it determines $a^\dag b$.
Finally, acting on $\rho$ with $r=q=0$, by
$\dot{q}=( b d^\dag-\half a^\dag b)s$ it determines $b d^\dag$.
So the matrix $k$ is determined up to a phase as $k=h$.

For higher dimensions
use the block matrix form, as in the Appendix \ref{appblock},
with $\rank(A)=\rank(R)=n-1$, where $S=s$ and $D=d$ are numbers, $B=|b\rangle$
and  $Q=|q\rangle$ are vectors.
The action on a $\rho$ with $|q\rangle=0$ and $s=0$ is $\dot{R}=\D_A(R)$ and
-- there is the induction hypothesis -- this determines $A$ up to a phase.
Note that $\dot{R}\neq 0$, so $A\neq a\one$.
The action on this $\rho$ is also $\dot{s}=0$, determining $C=0$
and $|\dot{q}\rangle= RA^\dag|b\rangle$.
Choosing a phase for $A$, the action on $\rho$ with
$R=0$, $|q\rangle=0$ determines $|b\rangle$ via $\dot{R}=|b\rangle s\langle b|$
and via $|\dot{q}\rangle$ from before;
and it determines $d$ via $|\dot{q}\rangle=s(d^\dag |b\rangle -\half A^\dag |b\rangle)$
in case $|b\rangle\neq 0$.
If $|b\rangle=0$, consider the action on a $\rho$ with $|q\rangle\neq 0$:
$|\dot{q}\rangle= d^\dag A|q\rangle -\half[A^\dag A+|d|^2]|q\rangle$.
This is a quadratic equation for the number $d$, but there are at least two different
equations for the components of $|q\rangle$ to determine $d$.
So the proof of uniqueness is completed for nondiagonalizable $h$.

\end{proof}

If $h$ is a diagonal operator,
all the matrix elements of $\rho$ evolve in time independently of each other,
$$\dot{\rho}_{i,j} =[h_i h_j^\dag - \half(|h_i|^2+|h_j|^2)]\rho_{i,j},$$
the diagonal of $\rho$ is invariant in time.
Only diagonal $k$ in $\D_k$ can give the same evolution.
But there are not always enough equations
to determine the matrix elements on the diagonal of $k$.
There remains some freedom in choosing $k$ with
$|k_i-k_j|=|h_i-h_j|$ and $Im(k_i k_j^\dag)=Im(h_i h_j^\dag)$
in special cases, e.g. if there are not more than two different values of $h_j$,
or if all $h_j$ are real.
But the qualitative structure of possible $k$ is fixed:
It is determined that it is a diagonal operator in the same basis,
with the same degeneracies.

\subsection{At the boundary of the set of states}\label{boundary}

It turns out that the time evolution of $\rho$ in the set $\s$ can be analysed
by studying the action of $\D$ on the boundary, the set of
density matrices with rank less than the dimension of the Hilbert space.
Referring to the inequality (\ref{edecrease}), we note
that it implies one of the central observations:

\begin{lem}{\bf No purifying in finite time:}\label{nopure}
The rank of $\rho(t)$ cannot decrease.
\end{lem}

An observation of almost sure increase of the rank
serves as a companion to this exclusion of a decrease.

\begin{lem}\label{firstorder}
Consider $\rho(t)$ with $\rank(\rho(0))<\dim(\Hh)$.
At least one eigenvalue which is zero for $\rho(0)$ becomes strictly positive
in first order in t, if and only if
the range of $\rho(0)$ is not an eigenspace or a generalized eigenspace of $h$.
\end{lem}

\begin{proof}
Consider the analytic expansion of $\rho(t)$ and its eigenvalues $r(t)$,
using the eigenvectors related to them, and
consider the terms linear in $t$.
$$\dot{r}(0)=\langle\psi|\dot{\rho}(0)|\psi\rangle$$
Write $h$ in Schur triangulated form, such that the upper left part $R$
(we use here the notation of the Appendix \ref{appblock}) acts
as a matrix of full rank in the range of $\rho(0)$.
Use the block matrix equation (\ref{rankplus}), with $R$ the restriction of $\rho$ to its range, $\phi$
an eigenvector to the eigenvalue zero, so it is not in the range of $R$ and
\beq
\dot{r}(0)=\langle\phi|\dot{S}(0)|\phi\rangle=\langle\phi|C\, R\,C^\dag|\phi\rangle.
\eeq
Observe that, according to (\ref{eigen}), the range of $\rho(0)$ is a (generalized) eigenspace, iff $C=0$.
Since $R$ is of full rank there is at least one $\phi$ giving a strict increase
$\dot{r}>0$, iff $C\neq 0$.
\end{proof}

Note the \emph{correspondence} between \emph{subspaces} of the Hilbert space $\Hh$
and  \emph{faces} of the set of states $\s$: Each face consists of
density matrices whose range is some subspace of $\Hh$.
Since not all subspaces of $\Hh$ can be eigenspaces of a non-constant operator,
the corresponding faces form a zero-set in the boundary of the set of states.
Here comes the structure of the operator $h$ into play.

\begin{definition}\label{lazy}
We define each face which corresponds to an eigenspace
or a generalized eigenspace of $h$ as a \textbf{lazy face}.
In a formal notation: If $P$ is the projector onto the range of $\rho$,
and if $P$ fulfills the property (\ref{eigen}) i.e. $hP=PhP$, then, and only then, is $\rho$ in a lazy face.
\end{definition}
The reason for this definition is Lemma \ref{firstorder},
and its strengthening in the following Lemma \ref{secondorder}.
Eigenvalues $r(t)$ which are zero for $t=0$
do not increase in first, not even in second order in $t$ if $\rho(0)$ is in a lazy face.
They may increase in third or a higher order, or not at all.
Examples will be presented in Subsection \ref{twobytwo}.
\begin{lem}\label{secondorder}
Consider $\rho(t)$ with $\rho(0)$ in a lazy face.
Each eigenvalue $r(t)$ beginning as $r(0)=0$ does not increase
in first and not in second order in $t$.
\end{lem}
\begin{proof}
The eigenvalue equation $\rho(t)\psi(t)=r(t)\psi(t)$ has to hold
in each order of $t$.
Consider the splitting $\psi=\chi\oplus\phi$, with $\chi(0)=0$, the block matrix form
presented in Appendix \ref{appblock}, with $C=0$, $Q(0)=0$, $S(0)=0$, and the Taylor expansions.
For the blocks of $\rho$ we need
\beq\label{rhotaylor}
R(t)=R+O(t),\quad Q(t)=-\half RA^\dag B\,t+O(t^2),\quad S(t)=\half B^\dag ARA^\dag B\,t^2+O(t^3).
\eeq
The first order of the eigenvalue equation, $\dot{\rho}(0)\psi(0)+\rho(0)\dot{\psi}(0)=\dot{r}(0)\psi(0)$,
gives two equations, one for the vector component orthogonal to the range of $\rho(0)$,
$$-\half B^\dag AR\chi(0)=\dot{r}(0)\phi(0)\quad\Rightarrow\quad \dot{r}(0)=0,$$
and one for the component  in the range of $\rho(0)$
$$\quad -\half RA^\dag B\phi(0)+R\dot{\chi}(0)=0\quad\Rightarrow
\quad \dot{\chi}(0)= \half A^\dag B\phi(0).$$
The last implication uses the invertibility of $R$, due to its full rank.
From the second-order equation we need only the part in the subspace orthogonal to the range of $\rho(0)$:\quad
$\half\ddot{\rho}(0)\psi(0)+\dot{\rho}(0)\dot{\psi}(0)+\half\rho(0)\ddot{\psi}(0)=\half\ddot{r}(0)\psi(0)$
gives
$$\frac14 B^\dag ARA^\dag B \phi(0)-\half B^\dag AR\dot{\chi}(0)+0=\half\ddot{r}(0)\phi(0)
 \quad\Rightarrow\quad \ddot{r}(0)=0,$$
since the terms on the left hand side cancel when the identity for $\dot{\chi}(0)$ is used.
\end{proof}

There is a way to geometrically characterize the action of $\D$ on the states in a face.
Define the mid point $\omega_f$ of the face as
$$\omega_f=\frac1n \sum_j|\phi_j\rangle\langle\phi_j|,$$
where the $\phi_j$ form a basis of the $n$-dimensional subspace corresponding to the face.
Consider the projection of the path $\rho(t)$ in the Euclidean space of hermitian matrices
onto the line which connects $\omega_f$ with the
maximally mixed state $\omega$ in the center of $\s$.
This projection gives a special coordinate for the HS vector pointing from $\omega_f$ to $\rho(t)$:
\beq
\Tr[(\rho(t)-\omega_f)(\omega-\omega_f)]=\frac1n\Tr[S(t)].
\eeq
Only for lazy faces does this coordinate not grow linearly in $t$.
It grows quadratically in $t$, see (\ref{rhotaylor}), so this geometrical characterization
is not as strong as the statements on the eigenvalues made in Lemma \ref{secondorder}.

In the next Section we explain in detail how an appearance of zero as an eigenvalue for $h$ has consequences
in the investigations on the set of invariant states.
They can be elements of the corresponding lazy face.
Concerning the dimension of such a face we state as a remark the following
\begin{lem}
An indecomposable operator $h$ of rank $n$ can not have more than
$n/2$ proper eigenvectors to a degenerate eigenvalue.
\end{lem}
\begin{proof}
Write $h$ in triangulated form and use the $m$ proper eigenvectors to
the degenerate eigenvalue $\lambda$ as the first $m$ basisvectors.
$$h=\left( \begin{array}{cc} A & B \\ 0 & D \end{array} \right),$$
where $A=\lambda\one_m$. The rectangular matrix $B$ can not have more rows than columns,
i.e. $m\leq n-m$; otherwise the rows could not be linearly independent
and there would be a unitary transformation of the m-dimensional
proper eigenspace such that $m-(n-m)$ rows of the transformed $B$ would
be zero, contradicting the undecomposability of $h$.
\end{proof}

\subsection{Stationary states}\label{invariant}

The upshot of the investigations on the action of $\D_h$ at the border are the following two
complementary propositions:
\begin{proposition}\label{nozero}{\bf Quitting the boundary.}
$\D_h$ has no stationary state at the boundary
if $h$ is indecomposable and does not have zero as an eigenvalue.
\end{proposition}
\begin{proof}
Consider $\rho$ at the boundary, write it as a block matrix
(notation as in the Appendix \ref{appblock}) with $R$ a matrix of full rank,
i.e. $R$ has no zero-eigenvalues, $Q=0$ and $S=0$.
If $C\neq 0$, use equation (\ref{rankplus}), which implies $\dot{S}\neq 0$. Otherwise,
apply equation (\ref{qdot}): $\dot{Q}=RA^\dag B \neq 0$,
since neither $R$ nor $A^\dag$ have zero eigenvalues and $B\neq 0$ because of the non-splitting property of $h$.
\end{proof}

But otherwise, one of the lazy faces may be an \textbf{attractive face}:
\begin{proposition}\label{zerosub}{\bf Eigenvalues zero are attractive.}
If $h$ is indecomposable and has zero as an eigenvalue, then the proper eigenspace
to this eigenvalue zero corresponds to a
face with stationary states.
This face is an attractor, and no other state is stationary.
\end{proposition}

\begin{proof}
If $\rho$ is in this face, obviously $h\rho=0$,  $\rho h^\dag=0$ and
therefore $\D_h(\rho)=0$ hold. On the other hand, if $\rho$ is not in this face,
one analyzes $P^\perp\rho P^\perp$, where $P$ is the projector onto the proper eigenspace
belonging to the eigenvalue zero of $h$.
Use the triangulated block matrix form
presented in Appendix \ref{appblock}, with $A =0$ acting on the zero-subspace, with
$C=0$, and $S=P^\perp\rho P^\perp$.
Equation (\ref{sdot}) reduces to $\dot{S}= \D_D (S) -\half [B^\dag BS+SB^\dag B]$,
and $\frac d{dt}\Tr[S]= -\Tr[BSB^\dag]$.
Looking at $S_\infty =\lim_{T\to\infty} \frac 1T\int_0^T S(t)dt$, which must be stationary in time,
we find that both terms have to vanish separately in
$\dot{S}_\infty =\D_D (S_\infty) -\half [B^\dag BS_\infty+S_\infty B^\dag B]=0$.

First, assume $h$ does not have improper eigenvectors to the eigenvalue zero.
This means that all the zeros in the diagonal of the triangulated $h$ appear in $A$
and $D=P^\perp h P^\perp$ is without zero-eigenvalues.
Since $h$ is indecomposable, $B\neq 0$.
If $\rho$ has $S\neq 0$, then, as is stated in proposition \ref{nozero}, $\D_D(S)\neq 0$
unless $S$ is a matrix of full rank without zero eigenvalues; but then
$B^\dag BS + SB^\dag B\neq 0$, acting as emptying out the subspace
which is orthogonal to the zero-subspace, and leading to $S=0$.

Now it remains to consider $h$ with improper eigenvectors to the eigenvalue zero.
These zeros appear in the diagonal of $D=P^\perp h P^\perp$,
at least one of them as a proper eigenvalue of $D$,
which acts in $P^\perp\Hh$, a space of lower dimension than that where $h$ acts.
We make an induction on the dimension of the Hilbert space.
The induction hypothesis is, that the sub-subspace $\mathcal K$ spanned by the proper eigenvectors
for the zeros of $D$ makes an attractive face for $\breve{S}(t)$,
which is defined by $\frac d{dt}{\breve{S}}= \D_D (\breve{S})$,
one part of the time evolution of $S(t)$.
The stationary $S_\infty$ vanishes on $\Hh -\mathcal K$.
%%The orthogonal subspace $P^\perp (\Hh-\mathcal K)$ is emptied out in the course of time,
%%if not by the action of $B^\dag BS + SB^\dag B$, then by the action of $\D_D$.
But the sub-subspace $\mathcal K$ is an improper eigenspace for $h$,
for each $|\psi\rangle\in \mathcal K$ one has $B|\psi\rangle\neq 0$.
So $\frac d{dt}\Tr[S(t)]=-\Tr[BS(t)B^\dag]\neq 0$
for each $S$ which does not vanish on $\mathcal K$, and $P^\perp\Hh$
is emptied out in the course of time.
\end{proof}

\begin{thm}\label{simplestat}{\bf All the stationary states for simple generators.}
\begin{description}
\item[a)]\textbf{Dissipation.} If $h$ is indecomposable and has no zero-eigenvalues, there is a unique stationary state in the interior.
It is
\beq\label{invrho}
\rho_\infty=(h^\dag h)^{-1}/\Tr[(h^\dag h)^{-1}]
\eeq
\item[b)]\textbf{Decay.}  If an indecomposable $h$ does have zero-eigenvalues,
the stationary states form an attractive face,
corresponding to the space spanned by the zero-eigenvalues.
\item[c)]\textbf{Elementary dephasing.}  If, in the other extreme, $h$ is diagonalizable,
the set of invariant states is the set of density matrices commuting with $h$.
\item[d)]\textbf{Stationary splitting.}  If $h$ is decomposable, $h=\bigoplus_j h_j$, each $h_j$ indecomposable,
the direct sums of stationary states $\rho_j$ of its parts
form the set of all the stationary states which do not have phase relations
between the independent subspaces $\Hh_j$.
\item[e)]\textbf{Dephasing of parts.}  If $h$ is decomposable,
there is a set of invariant phase relations $\rho_{j,\ell}$ between
the subspaces which are domains of definition of $h_j$ and $h_\ell$,
if and only if either both $h_j$ and $h_\ell$ have zero-eigenvalues,
or  $h_j \cong h_\ell$.
In the first case the invariant off-diagonal block matrices are
$\rho_{j,\ell}=\sum_{\alpha,\beta}c_{\alpha,\beta}|v_{j,\alpha}\rangle\langle w_{\ell,\beta}|$
where the $|v_{j,\alpha}\rangle$ and $|w_{\ell,\beta}\rangle$ are the proper zero-eigenvectors of $h_j$ and $h_\ell$.
In the second case any invariant phase relation $\rho_{j,\ell}$ is equivalent
to $\rho_j$ multiplied with some complex number,
where $\rho_j\cong \rho_\ell$ is the stationary state of $\D_{h_j}$.
\end{description}
\end{thm}
\begin{proof}

a) The stationarity of $\rho_\infty$ has been stated and proved in Proposition \ref{existence}
and in its proof.
Suppose there were two different stationary states, $\rho_0$ and $\rho_\infty$.
By linearity of the evolution equation, the whole line $\rho_0+\lambda(\rho_\infty-\rho_0)$
would consist of invariant matrices.
But this line would intersect the boundary where there are no stationary states,
see Proposition \ref{nozero}.

b) See Proposition \ref{zerosub}.

c) That $[h,\rho]=0$ implies $\D_h(\rho)=0$ is obvious.
  The off-diagonal elements obey $\dot{\rho}_{i,j}=(h_i^*h_j-\half(|h_i|^2+|h_j|^2))$,
  which is not zero, if $h_i\neq h_j$.
   For details on the decrease of phase relations $\rho_{i,j}$ in this case
   see Subsection \ref{phases}.

d) Thats the combination of the previous results, using Proposition \ref{simplesplit}.

e) Any density matrix $\rho$ for the whole Hilbert space $\Hh$ can be
considered as a big block matrix with entries $\rho_{j,\ell}$,
corresponding to the splitting of $\Hh =\bigoplus_j\Hh_j$ defined in Proposition \ref{simplesplit}.
Each $\rho_{j,\ell}$ has a time evolution on its own, without any mixing with the other blocks.
The stationary diagonal elements $\rho_{j,j}$ have been identified in a) to d).
Knowledge about them, together with the general preservation of positivity,
helps to find possible invariant off-diagonal blocks  $\rho_{j,\ell}$.
So, for any pair $h_j$, $h_\ell$ of indecomposable parts of $h$, we look at
a state $\rho$ restricted to the subspace $\Hh_j \oplus\Hh_\ell$.

First assume that $h_j$ has zero as an eigenvalue.
Since positivity is preserved as $t\to\infty$, and since
stationary $\rho_{j,j}$ reside on the subspace spanned by zero-eigenvectors,
an invariant  $\rho_{j,\ell}$ must also obey $h_j\rho_{j,\ell}=0$.
Its time evolution is  $\dot{\rho}_{j,\ell}=-\half\rho_{j,\ell}h_\ell^\dag h_\ell$,
and this can vanish only if  $\rho_{j,\ell}h_\ell=0$.
Invariant phase relations exist therefore only if also $h_\ell$ does have
zero eigenvectors. The stationary $\rho_{j,\ell}$ can then be expanded in ket-bra products
of zero-eigenvectors.

Now consider the case that both $h_j$ and $h_\ell$ are indecomposable and do not have
zero as an eigenvalue.
Assume that an invariant  $\rho_{j,\ell}\neq 0$ exists. Keeping the diagonal blocks
 $\rho_{j,j}$ and $\rho_{\ell,\ell}$ fixed, given by the unique stationary states
 determined by $h_j$ and $h_\ell$, multiplied with any non-zero weights,
the off-diagonal elements can be chosen with some factor $\lambda$,
 as $\lambda\rho_{j,\ell}$ and $\rho_{\ell,j}= \lambda^*\rho_{j,\ell}^\dag$,
 such that the whole $\rho$ is at the boundary of the set of states.
 Now we refer to the formulas stated in Appendix \ref{appblock}, formula
(\ref{rankplus}) and the following lines.
As is shown there, the density matrix, now represented by $R$ acting on a subspace,
is invariant only if $C=0$ and also $RA^\dag B=0$.
Now $R$ has no zero-eigenvalues, and $A$ has a subset of the whole set
of diagonal elements of $h_j\oplus h_\ell$ as eigenvalues, also without zeros.
So $B=0$ must hold.
But then $h_j\oplus h_\ell\cong\left( \begin{array}{cc} A & 0 \\ 0 & D \end{array} \right)$
commutes with the projector $P=\left( \begin{array}{cc} 1 & 0 \\ 0 & 0 \end{array} \right)$.
Since neither $h_j$ nor $h_\ell$ commute on their domains of definition $P_j\Hh$ and $P_\ell\Hh$ with
anything else but the constant operators, the only possibility for such a
nontrivial $P$ is the unitary equivalence of $h_j$ with $h_\ell$:
$h_j=Vh_\ell V^\dag$, \quad $V\cdot V^\dag=\one_j$,
$$P= p\one_j + \sqrt{p(1-p)}V + \sqrt{p(1-p)}V^\dag + (1-p)\one_\ell. $$
This unitary equivalence implies then the equivalence of the invariant blocks
with some factors, equivalence of
$c_j\cdot\rho_{j,j}$ with $c_{j,\ell}\cdot\rho_{j,\ell}$ and $c_\ell\cdot\rho_{\ell,\ell}$.
\end{proof}

We may reformulate the results, so that equivalent parts are
collected as factors in a tensor product:
If one knows the stationary states $\rho_{j,\alpha}$ for an indecomposable $h_j$,
one knows all possible states for $\one\otimes h_j$; they are $\sum_\alpha\sigma_\alpha\otimes \rho_{j,\alpha}$,
for any positive matrix $\sigma_\alpha$ which has $\Tr[\sigma_\alpha]=0$.
Stationary phase relations between the parts may exist.

In the case of an indecomposable $h$ with no zero-eigenvalues it can be shown,
as a companion to Proposition \ref{zerosub}, that
the unique stationary state is an attractor.
See Subsection \ref{nocircle}.

Parts a) and b) of Theorem \ref{simplestat} and the formula (\ref{invrho})
allow for an inversion, constructing an operator $h$ in such a way that the
evolution generated by $\D_h$ has to lead to a given state $\rho$.

\begin{thm}{\bf Construction of a simple evolution leading to a given state.}\label{leading}
For $\rho$ either a pure state or a state with a density matrix of full rank,
but non constant, there exist dynamical semigroups with a simple GKS-Lindblad
generator which have the given $\rho$ as the unique stationary state.
\end{thm}
\begin{proof}
In case $\rho$ is a pure state, take an indecomposable $h$ with just one zero eigenvalue,
so that $h\rho=0$. This may be, for example, $h=\sum_j|\phi_j\rangle\langle\phi_{j+1}|$,
with $|\phi_0\rangle\langle\phi_0|=\rho$.

In case $\rho$ is a matrix of full rank, construct $h=U|h|$ with $|h|=\rho^{-1/2}$,
as it has to be according to formula (\ref{invrho}).
This $h$ is decomposable iff there exists a nontrivial projector $P$ commuting with $h$.
This commutation implies also the commutation of $P$ with $h^\dag$, with $|h|^2=h^\dag h$ and with $|h|$.
Since this $|h|$ has no zero eigenvalue and is invertible, $P$ has also to commute with $U$.
As commuting operators, $P$ and $|h|$ have common eigenvectors $\phi_j$, such that
$|h|=\sum_j r_j^{-1/2}|\phi_j\rangle\langle\phi_j|$, and
$P=\sum_{j\in I} |\phi_j\rangle\langle\phi_j|$.
Now one has to consider a unitary operator $U$ which
mixes all the eigenspaces of $|h|$ belonging to different eigenvalues.
This may be, for example, $U=\sum|\phi_j\rangle\langle\phi_{j+1}|+|\phi_N\rangle\langle\phi_{0}|$,
where $N$ is the dimension of the Hilbert space.
Such a $U$ does not commute with
any $P$ of the sort which commutes with $h$,
and $h$ is indecomposable.
\end{proof}

For exceptional $\rho$, not pure but not invertible, and also for the maximally
mixed $\omega$, it is possible to find a semigroup with this $\rho$ as
a stationary state, however it will not be the unique one.

\subsection{The geometry of paths}\label{nocircle}

The geometric characterization of the paths $\rho(t)$ is equivalent to
the algebraic characterization of the superoperator $\D$,
determining its eigenvalues, its proper eigenspaces and also the generalized eigenspaces.
Logical arguments can work in both directions, from geometry to algebra or vice versa.

On the geometrical side we know that positivity of $\rho$ is preserved,
that the compact set $\s$ of states is mapped into itself.
The algebraic consequence for $\D$ is that there is no eigenvalue with strictly positive real part.
The detailed argument for this fact can be seen in the following discussion
which is needed for deeper analysis.

The reflection $\sigma\leftrightarrow\sigma^\dag$ is compatible with the action of $\D$,
i.e. $\D(\sigma^\dag)=[\D(\sigma)]^\dag$.
The consequences are: If an eigenvalue $\lambda$ of $\D$ is a real number,
the accompanying eigenmatrix and generalized eigenmatrices -- if there exist any -- can be chosen as selfadjoint.
Each pair of a complex eigenvalue $\lambda$ together with the eigenmatrix $\sigma$
has a mirror companion in the pair $\lambda^*$ and  $\sigma^\dag$.
For the set of self adjoint matrices this implies that for each complex number $z$
the matrix $\tau_z=z\sigma+z^*\sigma^\dag$ lies on a path $\tau_z(t)$ which is a spiral.
The imaginary part of $\lambda$ determines the time $T$ of revolution,
the real part gives the change of the HS-norm.
\beq
T:=2\pi/Im(\lambda),\qquad \tau_z(T)=e^{2\pi Re(\lambda)T}\tau_z
\eeq
Now consider any $\rho$ in the interior of $\s$, and some number $\eps$,
such that
\beq\label{spiral}
\rho_\pm =\rho\pm\eps\tau_z
\eeq
are both inside of $\s$.
Their distance changes exponentially in time,
$$\|\rho_+(t)-\rho_-(t)\| =\|\rho_+(0)-\rho_-(0)\|\cdot e^{2\pi Re(\lambda)t}.$$
But both of them remain in $\s$, their distance can not increase, $Re(\lambda)\leq 0$.

It remains to discuss the cases of $Re(\lambda)=0$.
We know for sure that stationary states exist.
They are eigenmatrices to the eigenvalue $\lambda=0$.
We state the absence of other cases as

\begin{thm}{\bf No circular paths.}\label{nocirclethm}
There are no circular paths in the HS-space of matrices;
there are no eigenvalues $\lambda$ of $\D$ on the imaginary axis except $\lambda=0$.
\end{thm}

\begin{proof}
We decompose the Hilbert space $\Hh=\bigoplus_j P_j \Hh$ and the operator $h=\sum_jP_jh P_j$
into indecomposable parts, as stated in Proposition \ref{simplesplit}.
For matrices $\sigma$ we use the block matrix decomposition $\sigma=\sum_{j,\ell}P_j\sigma P_\ell$
and observe the mutual independence of evolutions for the blocks,
$\D(P_j\sigma P_\ell)=P_j\D(\sigma) P_\ell$.
So the search for eigenmatrices can be done by considering single block matrices.

First we consider $h_j$ having zero as an eigenvalue.
We know by Proposition \ref{zerosub} that all the states $P_j\rho P_j$
move to the attractive face, where nothing  changes in time any more.
Now consider the path $\rho_+(t)$ starting with $\rho_0+\eps\tau_+$, as above in equation (\ref{spiral}).
$\rho_0(t)$ moves to a $\rho_\infty$ on the attractive face, as $t\to\infty$.
Also $\rho_+(t)$ has to stay inside of $\s$, moving to the attractive face,
where there can be no circular path.
So also $\rho_+(t)$ has to approach $\rho_\infty$, so $Re(\lambda)<0$, strictly.
The path may have the structure of a tornado.

For a matrix $\sigma=P_j\rho P_\ell$, giving the phase relation between $\Hh_j$ and $\Hh_\ell$,
we know that it has to move to the subspace of matrices obeying $h_j\sigma=0$,
because of preserving positivity for states on the domain $\Hh_j\oplus\Hh_\ell$.
For such a $\sigma$ the time evolution is dictated as $\dot{\sigma}=-\half\sigma h^\dag_\ell h_\ell$.
This gives either exponential decrease in time, or invariance, or a combination of both,
but no circular path, no nonzero eigenvalue of $\D$ on the imaginary axis.

Now it remains to consider those $\sigma=P_j\sigma P_j$ and $\sigma=P_j\sigma P_\ell$
where neither $h_j$ nor $h_\ell$ do have zero as an eigenvalue.
Here we turn to study the adjoint $\D^\dag$. It determines the time evolution in
the Heisenberg picture. Its set of eigenvalues is the complex conjugated set of
eigenvalues of $\D$, so these sets are actually identical.

We use the Kadison inequality, \cite{K52} (See the Appendix \ref{cpmaps}),
\beq
\Phi_t(F^\dag F)\geq \Phi_t(F^\dag)\Phi_t(F)
\eeq
which holds for the mapping $F\mapsto \Phi_t(F)$ of time evolution in the Heisenberg picture.

Consider an eigenmatrix $F=P_\ell\, F\, P_j$, with $\D^\dag (F)=\lambda F$.
The adjoint $F^\dag$ is an eigenmatrix with eigenvalue $\lambda^*$
Assume $\lambda=i\, r$ with $r\in \R$.
Then $ \Phi_t(F^\dag)\Phi_t(F)=(e^{-irt}F^\dag)(e^{irt}F)=F^\dag F$.
The matrix $F^\dag F$ maps $\Hh_j$ into $\Hh_j$ and is positive.
Take the projector onto an eigenvector with the largest eigenvalue, which is $\|F^\dag F\|$, as
a density matrix $\rho$, and switch for the moment between the Schr\"{o}dinger and the Heisenberg pictures:
\beq
\Tr[\rho(t)\,F^\dag F]=\Tr[\rho\,\Phi_t(F^\dag F)]\geq\Tr [\rho\, F^\dag F]=\|F^\dag F\|.
\eeq
The mean value over the time $\geq 0$ has to give the stationary state $\bar{\rho}$ on $\Hh_j$,
which is unique and has no zero eigenvectors in $\Hh_j$, see the Theorem \ref{simplestat}.
\beq
\Tr[\bar{\rho}\,F^\dag F]= \lim_{T\to\infty}\frac1T \int_0^T\Tr[\rho(t)\,F^\dag F]\geq\|F^\dag F\|.
\eeq
On the other hand, obviously $\Tr[\bar{\rho}\,F^\dag F]\leq \|F^\dag F\|$,
so $\Tr[\bar{\rho}\,F^\dag F] = \|F^\dag F\|$.
In a basis where $\bar{\rho}$ is diagonal, the matrix representing $F^\dag F$
has therefore all diagonal elements equal to its norm, which implies vanishing of
all the off-diagonal elements, so $F^\dag F= \|F^\dag F\|\one_j$.

If $\ell\neq j$ one finds also $F\, F^\dag= \|F\, F^\dag\|\one_\ell$ by doing
analogous investigations. A first consequence is, that the existence of such an eigenmatrix
is possible only if $\Hh_j$ and $\Hh_\ell$ have equal dimension.
With $F\,F^\dag\, F=F\,\|F^\dag\, F\|=\|F\,F^\dag\|\, F$
we conclude $\|F^\dag\, F\|=\|F\,F^\dag\|$, so $V:=F/\sqrt{\|F^\dag F\|}$ is
a unitary operator if $\Hh_j=\Hh_\ell$; and it is an isometry between
$\Hh_j$ and $\Hh_\ell$ if these are different subspaces.
We turn back to the Heisenberg picture and multiply the evolution equation (\ref{Heisenberg})
for $V$ from the left by $V^\dag$:
\beq
V^\dag \D^\dag(V)=V^\dag h^\dag_\ell\, V \,h_j-\half V^\dag h^\dag_\ell h_\ell V -\half V^\dag V h_j^\dag h_j
=V^\dag\lambda V=i\, r \one_j \,.
\eeq
We insert $V\,V^\dag$ between $h^\dag_\ell$ and $h_\ell$, define $\hat{h}_\ell=V^\dag h_\ell V$,
take the trace and read it as an equation for inner products in the HS-space:
\beq\label{schwarzplus}
\langle \hat{h}_\ell|h_j\rangle =\half(\langle \hat{h}_\ell|\hat{h}_\ell\rangle+\langle h_j|h_j\rangle)+i\, r\, n,
\eeq
where $n=\dim(\Hh_j)$. For $j=\ell$ we can conclude immediately that $\lambda=i\, r=0$.
For $j\neq\ell$ this follows from the Cauchy Schwarz inequality and the
inequality between the geometric and the arithmetic mean
\beq
|\langle \hat{h}_\ell|h_j\rangle|\leq \|\hat{h}_\ell\|\,\|h_j\|\leq\half(\|\hat{h}_\ell\|^2+\|h_j\|^2).
\eeq
\end{proof}

\section{Special systems and examples}\label{special}

\subsection{Dephasing}\label{phases}

The stable phase relations between indecomposable parts have been identified in Theorem (\ref{simplestat}).
In some special cases we can state details on the integrated evolution:
One may restrict the study to indecomposable parts.
We use the same notation as in Appendix \ref{appblock}, applied to the operators in the subspace
$\Hh_j\oplus\Hh_\ell$, writing $A=h_j$ and $D=h_\ell$ for the parts of $h$.
The phase relations are expressed in $Q=\rho_{j,\ell}$.
They evolve, according to equation (\ref{qdot}) as
\beq\label{phasedot}
\dot{Q}=AQD^\dag-\half(A^\dag A Q+ Q D^\dag D),
\eeq

\textbf{a)}
Looking at a diagonalizable $h$, with $A=a$ and $D=d$ just numbers,
the phase relations, the off-diagonal elements of $\rho$, evolve according to
equation (\ref{phasedot}) as
\beq
\frac d{dt} q = -\half (|a-d|^2+(a^* d-a\, d^*))q,
\eeq
giving exponential decrease,
$|q(t)|= \exp(-|a-d|^2)q(0),$
together with a phase rotation.

\textbf{b)}
Consider $D=h_\ell$ having zero-eigenvectors and $QD^\dag=0$.
Then $$\dot{Q}=-\half A^\dag AQ.$$
Expand $Q=\sum_\alpha c_\alpha Q_\alpha$ with $A^\dag AQ_\alpha=a_\alpha Q_\alpha$.
Exponential decrease is the consequence:

\beq
Q(t)=\sum_\alpha e^{-ta_\alpha/2}Q_\alpha.
\eeq

\textbf{c)}
Consider $\Hh_\ell$ as one-dimensional, with $h_\ell=d$, and the other part $h_j=A$ as indecomposable.
Then equation (\ref{phasedot}) implies monotone decrease of $\Tr(Q^\dag Q)$:
\beq
\frac d{dt} \Tr (Q^\dag Q) =-\Tr (Q^\dag[(A^\dag-d^\dag)(A-d)]Q)\leq 0.
\eeq
This inequality is strict, unless $d$ is an eigenvalue of $A$ and the columns of $Q$ are corresponding eigenvectors.
In case of $(A-d)Q=0$ this first order derivative of $\|Q\|^2$ vanishes. Nevertheless one gets
\beq\label{firstdiff}
\dot{Q}=(|d|^2-\half A^\dag d-\half|d|^2)Q=\half d(d^*-A^\dag)Q\neq 0,
\eeq
since $A$ is an indecomposable $h_j$, so $A$ and $A^\dag$ have no common eigenvalues.
There is no stationary phase relation.
The second order derivative of $Q^\dag Q$ in t is
also zero, but the third order derivative is strictly negative.

To make the calculation transparent, we consider $Q$ as an element $|Q\rangle$
of the Hilbert-Schmidt-space, with the inner product $\langle P|Q\rangle=\Tr (P^\dag Q)$.
Time-derivative is a super-operator $\D$, mapping $|Q\rangle$ to $|\dot{Q}\rangle$.
In the formulas we use the same letter $A$ for the special super-operator
which we define as $A|Q\rangle :=|A \, Q\rangle$,
as for the operator in $\Hh$.
\beq
|\dot{Q}\rangle=\D|Q\rangle=[-\half(A^\dag-d^*)(A-d)+\half(Ad^*-A^\dag d)]|Q\rangle .
\eeq
Its adjoint operator is
\beq
\langle \dot{Q}|=\langle Q|\D^\dag =\langle Q|[-\half(A^\dag-d^*)(A-d)-\half(ad^*-A^\dag d)] .
\eeq
In this formalism we get
$\frac d{dt} \Tr Q^\dag Q =\langle Q|(\D^\dag+\D)|Q\rangle$
and
\beq\label{timediff}
\frac {d^n}{dt^n} \Tr Q^\dag Q =\sum_{m=0}^n
\left( \begin{array}{c} n \\ m \end{array} \right)
\langle Q|\D^{\dag^m}\D^{n-m}|Q\rangle
\eeq
The assumption $AQ=dQ$ gives $(\D^\dag+\D)|Q\rangle=0$ and its adjoint,
$\langle Q|(\D^\dag+\D)=0$.
Now we write the second and first derivatives, using (\ref{timediff}), as
$$\frac {d^2}{dt^2} \Tr Q^\dag Q = \langle Q|\left((\D^\dag+\D)\D+\D^\dag(\D^\dag+\D)\right)|Q\rangle =0,$$
and, using $|Q\rangle\langle Q|/\|Q\|^2<\one$,
$$\frac {d^3}{dt^3} \Tr Q^\dag Q = \langle Q|
\left((\D^\dag+\D)\D^2+\D^\dag(\D^\dag+\D)\D+\D^{\dag^2}(\D^\dag+\D)\right)|Q\rangle$$
$$=\langle Q|\D^\dag(\D^\dag+\D)\D|Q\rangle =-\langle Q|\D^\dag(A^\dag-d^*)(A-d)\D|Q\rangle $$
$$\leq-\langle Q|\D^\dag(A^\dag-d^*)|Q\rangle\langle Q|(A-d)\D|Q\rangle/\|Q\|^2=
-\|\half d(d^*-A^\dag)Q\|^4/\|Q\|^2 < 0,$$
with a strict inequality, as in equation (\ref{firstdiff}).

We \emph{warn} of thinking to generalize these examples.
There is in general no monotone decrease of phase relations,
not in the form of a decrease of the HS-norm of $Q$.
Some remarks are stated in Subsection \ref{mixing}.

\subsection{Twodimensional systems}\label{twobytwo}

We present the stationary states, the investigations on the eigenvalues
and eigenmatrices of $\D$, and make remarks on the paths.

If $h$ is constant, then $\D_h$ is zero. This case is nevertheless worth mentioning.
In the way of decomposing a larger system it may appear as a restriction to a subspace.
The constant $h$ is then active in the evolution of the phase relations with other
subspaces.

For non constant $h$ consider the triangulated representation
\beq
h=\left( \begin{array}{cc} a & b \\ 0 & d \end{array} \right).
\eeq

\medskip $\bullet$
If $\textbf{b=0}$ there are four eigenmatrices of $\D_h$,
$$\left( \begin{array}{cc} 1 & 0 \\ 0 & 0 \end{array} \right),\quad
\left( \begin{array}{cc} 0 & 0 \\ 0 & 1 \end{array} \right),\quad
\left( \begin{array}{cc} 0 & 1 \\ 0 & 0 \end{array} \right),\quad
\left( \begin{array}{cc} 0 & 0 \\ 1 & 0 \end{array} \right),\quad
$$
the first two belonging to the eigenvalue $0$,
the other two to $\lambda=\half|a-d|^2+\half(a^*d-ad^*)$ and to $\lambda^*$.
Each diagonal $\rho$ is constant in time, dephasing occurs as stated in Subsection \ref{phases}, part a.

\medskip $\bullet$
For any $h$ with $\textbf{b}\neq \textbf{0}$
there is a unique normed stationary state
 $\rho$, with
\beq
\rho=\left(\begin{array}{cc} r & q \\ q^* & s \end{array} \right),\quad r=1-s,\quad
s=\frac{|a|^2}{|a|^2+|b|^2+|d|^2}\,,\quad q=\frac{-a^*b}{|a|^2+|b|^2+|d|^2}\,.
\eeq
This is easily seen by considering the equations (\ref{rdot},\ref{qdot},\ref{sdot})
in the simpler form with numbers instead of matrix-blocks;
and it is, of course, consistent with equation (\ref{invrho}).

Any kind of $\rho$ can appear as a stationary state for some $\D_h$, except the maximally mixed $\omega=\one/2$.
The stationary $\rho_0$ is a pure state if $a=0$ or $d=0$.

Concerning all the eigenvalues $\lambda$ of $\D_h$ and its eigenmatrices $\sigma=
\left(\begin{array}{cc} r & q \\ p & s \end{array} \right)$,
we note that the preservation of the trace takes the form $\dot{s}+\dot{r}=\lambda(r+s)=0$.
The case $\lambda=\lambda_0=0$ involves the unique stationary state.
The other eigenmatrices belong to $\lambda_i\neq 0$ and have $r+s=0$.
By using the equations (\ref{qdot},\ref{mdot},\ref{sdot}),
together with this condition $\Tr\sigma=0$, one gets
the eigenvalue equation $\dot{\sigma}=\D_h(\sigma)=\lambda\sigma$ in the form
\beq\label{threematrix}
\frac d{dt}\left( \begin{array}{c} q \\ p \\ s \end{array} \right)=
\left( \begin{array}{ccc} ad^*-\|h\|^2_2/2 & 0 & bd^* \\ 0 & a^*d-\|h\|^2_2/2  & b^*d \\
-\half ab^*  & -\half a^*b & -|b|^2  \end{array} \right)
\left( \begin{array}{c} q \\ p \\ s \end{array} \right)=
\lambda\left( \begin{array}{c} q \\ p \\ s \end{array} \right).
\eeq
The remaining eigenvalues of $\D$ are the zeros of
the characteristic polynomial for the $3\times 3$ matrix appearing in this equation.

\medskip $\bullet$
For $\textbf{a=0}$ there are only real eigenvalues.
They are $\lambda_1=-|b|^2$, $\lambda_2=\lambda_3=-\|h\|^2_2/2$,
the eigenmatrices are
\beq\label{eigenmatrices}
\sigma_1=\left( \begin{array}{cc} |b|^2-|d|^2 & 2bd^* \\ 2b^*d & |d|^2-|b|^2 \end{array} \right),\quad
\sigma_2=\left( \begin{array}{cc} 0 & 1 \\ 0 & 0 \end{array} \right),\quad
\sigma_3=\left( \begin{array}{cc} 0 & 0 \\ 1 & 0 \end{array} \right).
\eeq

For $a=0$ and $|b|=|d|\neq 0$ all three zeros of the characteristic polynomial coincide, and the eigenmatrix
$\sigma_1$ is no longer linearly independent of
$\sigma_2$ and $\sigma_3$. But there is now a generalized eigenmatrix
$$\hat{\sigma}_1=\left( \begin{array}{cc} 1 & 0 \\ 0 & -1 \end{array} \right),$$
satisfying the equation $\D(\hat{\sigma}_1)+|b|^2\hat{\sigma}_1=|b|^2(\sigma_2+\sigma_3)$.

\medskip $\bullet$
For $a\neq 0$, $b\neq 0$, $d\neq 0$ (note that $d=0$ is equivalent to $a=0$)
the three eigenvalues of $\D$ which are not zero
have to be calculated by finding the zeros of
the characteristic polynomial for the matrix appearing in (\ref{threematrix}), a polynomial of third order.
A simple case occurs
when $a$ and $d$ are real numbers and either $b$ or $bi$ is also real.
In these cases one real eigenvalue is $\lambda_1=-(|a-d|^2+|b|^2)/2$,
and the corresponding eigenmatrix $\sigma_1$ is, in case $b$ is real, the selfadjoint matrix
$\sigma_1=\left( \begin{array}{cc} 0 & i \\ -i & 0 \end{array} \right).$
In the case $b$ is imaginary it is
$\sigma_1=\left( \begin{array}{cc} 0 & 1 \\ 1 & 0 \end{array} \right).$

\medskip $\bullet$
Quitting the boundary, starting from the \textbf{lazy state}
$\rho(0)=\left( \begin{array}{cc} 1 & 0 \\ 0 & 0 \end{array} \right)$:
The eigenvalue $e(t)$ to the density matrix $\rho(t)$,
starting as $e(t)=0$, can be analyzed in its behavior
as an analytic function. Its Taylor expansion gives, as has been stated and proved in Lemma
\ref{secondorder}, $\dot{e}(t)=0$ and $\ddot{e}(t)=0$.
Continuing the procedure of solving the eigenvalue equation term by term according to
the order of $t^n$ gives: $$\frac{d^3}{dt^3}e(t)|_{t=0}=\half|a|^2|b|^2|a-d|^2\,.$$
Also this third order derivative vanishes if $a=d$, whereas, in this case
 $$\frac{d^4}{dt^4}e(t)|_{t=0}=\frac78|a|^2|b|^6.$$

\subsection{Mixing and demixing}\label{mixing}

The mixing property of a state can change, in the course of time, in any direction
for general $h$.
The case of dissipation can start with a pure state and will end up in a mixed state.
In case of decay, when $h$ has only one eigenvector to the eigenvalue zero,
each mixed state will decay and approach the pure decay product.
The path in $\s$ can also start from the boundary and go through the maximally mixed state $\omega$,
the following evolution will then act as purifying.
This means increase of entropy, followed by its decrease. Such cases have been found numerically in \cite{NTP07}
and in \cite{D04}.
If such an up-and-down process appears in parallel fashion
in two parts of the system, the phase relations between them can also show this non-monotonic behavior.

One may imagine a fast decay procedure with up-down evolution of entropy in a subspace,
followed by a slower procedure, going over to another subspace.
What is the end product of the first decay is then decaying again.
It is of course possible that these happenings repeat several times,
giving several ups and downs of entropy.
From the mathematicians point of view there is the question of forming
functions by superpositions of a finite number of exponential functions.
In reality these things happen in nuclear processes.

\section{Summary and Conclusion}\label{summary}

We consider the quantum dynamical semigroups representing
time evolutions of open systems with finite dimensional Hilbert space.
In this paper we concentrate on semigroups with a simple GKS-Lindblad generator.
An inequality concerning the decay of eigenvalues is recognized as
an important clue to the behavior of mixed states near the boundary of the set of states;
the rank can not decrease in finite time.
The complementary clue is a study on the increase of the rank of $\rho(t)$, eigenvalues
which start at zero getting positive.
There are few exceptions to an increase in first order $\sim t$.
For this analysis the characterization of structures is essential,
representing the operator $h$ which is used to build the generator $\D_h$
as an upper triangular matrix.
In relation to the subspaces appearing as eigenspaces to $h$ we define,
in one to one correspondence,
certain faces of $\s$, the convex set of states, as ``lazy'' faces,
since the rank of the $\rho$ in such a face does not increase in first order $\sim t$.

The three basic ways of non-invertible time evolution in physics,
i.e. dissipation, decay and decoherence/dephasing are mirrored
in the mathematical structures of $h$ and of the generator $\D_h$,
and the decomposition of $h$ is the basis for a complete characterization
of all the stationary states, including a concrete formula.
We have moreover characterized the dephasing as not allowing circular paths
of states, and we have characterized the stationary states as attractive.
With a view on applications in ``quantum engineering'' we stated
a possibility of inverting the task: find an evolution which leads to a
given state. This task can be fulfilled for each state with one exception:
The maximally mixed state.

Hereupon we demonstrated the findings on two-dimensional systems.

\section{Appendix}\label{appendix}

\subsection{Complete positivity}\label{cpmaps}

Dynamical maps of $\s$ into $\s$ have to preserve the trace.
The dual maps are then ``unital'' i.e. they map $\one$ to $\one$.

Stinespring's theorem states that $\phi$ mapping $\B(\Hh)$ to $\B(\Hh)$ is
completely positive and unital iff
\beqa
\exists\Hh_2,\quad \exists U=(U^\dag)^{-1}\,{\rm on}\,\,\Hh\otimes\Hh_2,\quad \exists\rho_2\in\s(\Hh_2)\nonumber\\
 \forall F:\quad
\Phi(F)=\Tr_2\left( (\one\otimes\rho_2)\cdot U^\dag \cdot (F\otimes\one)\cdot U\right) .\label{stine}
\eeqa
The predual map, acting on a state represented by the density matrix $\rho\in\B(\Hh)$ is therefore
\beq\label{cprho}
\rho\mapsto\Tr_2\left( U\cdot(\rho\otimes\rho_2)\cdot U^\dag \right).
\eeq
From a physicists point of view, (\ref{stine}) or (\ref{cprho}) are the relevant equations for a quantum dynamical map,
where  $\B(\Hh)$ is the algebra of observables for a system, and $\rho_2\in\s(\Hh_2)$ is the state of the ``reservoir''.
So we actually don't need the mathematicians definition of complete positivity.
In the sequel we use only (\ref{stine}) instead.

\begin{proposition}{\bf The Kadison inequality.}
If $\Phi$ is a unital completely positive map $\B(\Hh)\rightarrow\B(\Hh)$,
then for each $F\in\B(\Hh)$ there holds the inequality
\beq
\Phi(F^\dag F)\geq \Phi(F^\dag)\Phi(F)
\eeq
\end{proposition}
\begin{proof}
By doubling the ``reservoir''-space $\Hh_2$ to $\Hh_2\otimes\Hh_2$
one can consider the state $\rho$ in $\B(\Hh_2)$ as the restriction of an entangled pure state,
with a state vector $\Psi$ in $\Hh_2\otimes\Hh_2$.
Letting $U$ act on $\Hh\otimes\Hh_2\otimes\Hh_2$ one can reformulate (\ref{stine}) as
\beq
\Phi(F)=\langle\Psi |U^\dag \cdot (F\otimes\one\otimes\one)\cdot U|\Psi\rangle.\label{stinetwo}
\eeq
Now using $|\Psi\rangle\langle\Psi|<\one$ in $\Phi(F^\dag)\cdot\Phi(F)$ gives the Kadison inequality.
\end{proof}
If one feels uneasy with the abstract algebraic notation,
one may read (\ref{stinetwo}) as
\beq
\langle\phi|\Phi(F)|\psi\rangle=
\langle\phi\otimes\Psi |\,U^\dag \cdot (F\otimes\one\otimes\one)\cdot U\,|\psi\otimes\Psi\rangle,
\eeq
and use a set $|\phi_n \rangle$ of basis vectors in $\Hh$ to write the product $\Phi(F^\dag)\cdot\Phi(F)$ in
matrix notation, and insert $\sum_n|\phi_n\otimes\Psi\rangle\langle\phi_n\otimes\Psi|<\one$
between $\Phi(F^\dag)$ and $\Phi(F)$.

\subsection{Block matrices}\label{appblock}

Consider a splitting of $\Hh$ into subspaces, $\Hh=\Hh_1\oplus\Hh_2$. Write
vectors $\psi=\chi\oplus\phi$ as $\left( \begin{array}{c} \chi \\ \phi \end{array} \right)$
and write operators acting in $\Hh$ as block matrices:

\beqa\label{blockmatrices}
h=\left( \begin{array}{cc} A & B \\ C & D \end{array} \right) \\
\rho=\left( \begin{array}{cc} R & Q \\ Q^\dag & S \end{array} \right) \\
\textrm{The action of } \D_h :\quad \rho \quad\mapsto\quad \dot{\rho} =
\left( \begin{array}{cc} \dot{R} & \dot{Q}\\ \dot{Q}^{\dag}& \dot{S}\end{array} \right)
\eeqa
We need two special cases:

The first one is the case of density matrices $\rho$ with rank less than the dimension of the Hilbert space:

\beq\label{rankplus}
\textrm{If } Q=0\textrm{ and }S=0,\textrm{ then }\quad\quad\dot{S}=C\, R\, C^\dag.
\eeq
\bigskip
The second special case with no restriction on $\rho$ uses the Schur-Toeplitz triangulation, \cite{L69},
of the generating matrix $h$:
\bigskip\\If $C=0$, then
\beqa
\dot{R}&=& \D_A (R) + BSB^\dag +AQB^\dag +BQ^\dag A^\dag  -\half [QB^\dag A + A^\dag BQ^\dag ]\label{rdot} \\
\dot{Q}&=&AQD^\dag +BSD^\dag -\half [A^\dag A Q + QB^\dag B+QD^\dag D+RA^\dag B+A^\dag BS ] \label{qdot} \\
\dot{S}&=& \D_D (S) -\half [B^\dag BS+SB^\dag B \, + \, B^\dag AQ+Q^\dag A^\dag B].\label{sdot}
\eeqa
For a general matrix
$$\sigma=\left( \begin{array}{cc} R & Q \\ M & S \end{array} \right),$$
 which is in general not self adjoint, we
 have to replace $Q^\dag$ in (\ref{rdot}) and in (\ref{sdot}) by $M$, and we also need
 \beq\label{mdot}
\dot{M}=D M A^\dag +D S B^\dag -\half [M A^\dag A  + B^\dag B M+D^\dag D M+B^\dag A R +S B^\dag A ]
 \eeq

\end{document}